\documentclass[1p,final]{elsarticle}
\usepackage{amsfonts,natbib,pslatex}

\usepackage{latexsym,amsmath,amsthm,amssymb,amsxtra,mathrsfs,times,CJK}
\usepackage{makecell}
\usepackage{color}
\usepackage{amsfonts,color}
\usepackage{amssymb,amsmath,latexsym}
\usepackage{amssymb}
\usepackage{extarrows}

\usepackage{tikz}

\usepackage[colorlinks,linkcolor=blue,anchorcolor=blue,citecolor=blue]{hyperref}
\usepackage[para]{threeparttable}

\newtheorem{theorem}{Theorem}[section]

\newtheorem{definition}[theorem]{Definition}

\begin{document}

\begin{frontmatter}



\title{ Several classes of 0-APN power functions over $\mathbb{F}_{2^{n}}$}

\author[SWJTU]{Tao Fu}
 \ead{futao0923@163.com}

\author[SWJTU]{Haode Yan\corref{cor1}}
 \ead{hdyan@swjtu.edu.cn}

 \cortext[cor1]{Corresponding author}
 \address[SWJTU]{School of Mathematics, Southwest Jiaotong University, Chengdu, 610031, China}
%\address[NNU]{College of Mathematics and Statistics, Northwest Normal University, Lanzhou, 730070, China}
%\tnotetext[fn1]{}

% \author[label1,label2]{<author name>}
% \address[label1]{<address>}
% \address[label2]{<address>}

%\author[SWJTU]{Cuiling Fan\corref{cor1}}
% \ead{cuilingfan@163.com}
%\author[UB]{Nian Li}
% \ead{nianli.2010@gmail.com}
%\author[SWJTU]{Zhengchun Zhou}
% \ead{zzc@home.swjtu.edu.cn,zczhou@126.com}
%

% \address[SWJTU]{School of Mathematics, Southwest Jiaotong University, Chengdu, 610031, China}
%\address[UB]{Department of Informatics, University of Bergen, N-5020 Bergen, Norway}
%Department of Informatics, University of Bergen, N-5020 Bergen, Norway

\begin{abstract}
Recently, the investigation of Partially APN functions has attracted a lot of attention. In this paper, with the help of resultant elimination and MAGMA, we propose several new infinite classes of 0-APN power functions over $\mathbb{F}_{2^{n}}$. By the main result in \cite{CCZ}, these $0$-APN power functions are CCZ-inequivalent to the known ones. Moreover, these infinite classes of 0-APN power functions can explain some exponents for $1\leq n\leq11$ which are not yet ``explained"  in the tables of Budaghyan et al. \cite{LNCP}.	
\end{abstract}

\begin{keyword}
APN function \sep 0-APN  function \sep  Power function \sep Resultant
\MSC   11T06 \sep 94A60

\end{keyword}

\end{frontmatter}

%The title of your section 1
\section{Introduction}

Let $\mathbb{F}_{2^n}$ be the finite field with $2^n$ elements and $\mathbb{F}_{2^n}^*=\mathbb{F}_{2^n}\setminus \{0\}$, where $n$ is a positive integer.
Let $F$ be a function from  $\mathbb{F}_{2^n}$ to itself. The derivative function of $F$ with respect to any $a\in\mathbb{F}_{2^n}$ is the function $D_aF$ from $\mathbb{F}_{2^n}$ to $\mathbb{F}_{2^n}$ given by $D_aF(x)=F(x+a)+F(x)$, where $x\in \mathbb{F}_{2^n}$. For any $a\in\mathbb{F}_{2^n}^*$ and $b\in\mathbb{F}_{2^n}$, let  $N_{F}(a,b)$ denote the number of solutions $x\in \mathbb{F}_{2^n}$ of $D_aF(x)=b$. The differential uniformity of $F$ is 
define as $\Delta_F=\text{max}_{a\in\mathbb{F}_{2^n}^*,b\in\mathbb{F}_{2^n}}N_{F}(a,b)$.
We call the function $F$ differentially $k$-uniform if  $\Delta_F=k$. We expect $\Delta_F$ to achieve a smaller value when $F$ is used for resisting differential attacks. Note that the solutions of $D_aF(x)=b$ come in pairs, then the smallest $\Delta_F$ is $2$. When $\Delta_F=2$, $F$ is called an almost perfect nonlinear (APN for short) function. Currently known only 6 infinite classes of APN power functions over $\mathbb{F}_{2^n}$ are given in \cite{Welch}-\cite{Kasami}, \cite{Inverse}. 
%The recent progress on APN functions over $\mathbb{F}_{2^n}$ can be found in [] and their references.

%\cite{Gold},\cite{Dobbertin},\cite{Niho}

In order to study the conjecture of the highest possible algebric degree of APN functions, the concept of partially APN functions was first introduced by Budaghyan et al. \cite{LNSCP} as follows.

\begin{definition}(\cite{LNSCP})
	Let $F$ be a function from $\mathbb{F}_{2^n}$ to itself. For a fixed $x_{0}\in \mathbb{F}_{2^n}$, we call $F$ is $x_{0}$-APN (or partially APN) if all the points $u,v$ satisfying $F(x_{0})+F(u)+F(v)+F(x_{0}+u+v)=0$ belong to the curve $(x_{0}+u)(x_{0}+v)(u+v)=0$.
	\end{definition}

If $F$ is APN, it is obvious that $F$ is $x_0$-APN for any $x_0\in\mathbb{F}_{2^n}$. Conversely, there are many examples that are $x_0$-APN for some $x_0\in\mathbb{F}_{2^n}$ but not APN. In\cite{LNSCP}, Budaghyan et al. provided some propositions and characterizations of partial APN functions. Moreover, Pott proved that for any $n\geq3$, there are partial $0$-APN permutations on $\mathbb{F}_{2^n}$ in \cite{pott}. When $F$ is a power mapping, i.e.,  $F(x)=x^d$ for some integer $d$, due to the nice algebric structure of $F$, we only need to consider the partial APN properties of $F$ at $x_0=0$ and $x_0=1$. Moreover, $F$ is 0-APN if and only if the equation $F(x+1)+F(x)+1=0$ has no solution in $\mathbb{F}_{2^n}\backslash{\mathbb{F}_2}$. Many classes of 0-APN but not APN power functions over $\mathbb{F}_{2^n}$ are constructed in \cite{LNSCP} and \cite{LNCP}. They also list all power functions $F(x)=x^d$ over $\mathbb{F}_{2^n}$ for $1 \leq n \leq 11 $ that are 0-APN but not APN \cite{LNCP}. Recently, Qu and Li provided seven classes of 0-APN power functions over $\mathbb{F}_{2^n}$, some of them were proved to be locally-APN \cite{QL}. Very recently, Wang and Zha proposed several new infinite classes of 0-APN power function by using the multivariate method and resultant elimination \cite{WZ}. However, some pairs of $(d,n)$ in \cite{LNCP} are not yet ``explained", we summarize them in Table \ref{table1}. In this paper, we extend them to new infinite families of 0-APN power functions.
This paper is organized as follows. Some basic results of the resultant of polynomials are introduced in Section \ref{Preliminaries}. In Section \ref{Newclasses}, we propose seven new infinite classes of 0-APN power functions over $\mathbb{F}_{2^{n}}$. We mention that we follow the approach proposed in \cite{QL} and \cite{WZ}. For the convenience, we summarize the seven classes of 0-APN power functions in Table \ref{table2}. By the main result in \cite{CCZ}, the 0-APN power functions obtained in this paper are CCZ-inequivalent to the known ones.

% one can easily checked that $x^{2d}$ is a 0-APN function when $x^{d}$ is a 0-APN function. 

\begin{table}[!t]
\renewcommand\arraystretch{1.8}
\centering
\caption{0-APN power mappings $F(x)=x^d$ over $\mathbb{F}_{2^n}$ for $1\leq n\leq11$ which are not yet ``explained"}
\label{table1}       % Give a unique label
\begin{tabular}{|l|l|}
\hline    $n$ & $d$  \\
 \hline   9&45,125 \\
 \hline   10&51,93, 105, 351, 447 \\
   \hline   11&59,93,169,243,303,507, 245,447, 89, 445\\         
\hline
\end{tabular}
\end{table}

\begin{table}[!t]
	\renewcommand\arraystretch{1.8}
	\centering
	\caption{New classes of 0-APN power functions $f(x)=x^d$ over $\mathbb{F}_{2^n}$}
	\label{table2}       % Give a unique label
\begin{tabular}{|l|l|l|l|}
	\hline    $x^d$ & conditions & $(d,n)$ can be explained in Table \ref{table1} &Reference  \\
	 \hline	$x^{3\cdot2^{k}-7}$&$n=2k+1$&$(89,11)$&Thm \ref{TH6} \\ 
	\hline	$x^{2^{2k+1}-2^{k+1}-2^k+1}$&$n=3k+1$& $(105,10)$&Thm \ref{TH4}\\ 
	\hline   $x^{3(2^k-1)}$&$n=2k,3\nmid k$& $(93,10)$&Thm \ref{TH2} \\
	\hline	$x^{5(2^{k+1}+2^k+1)}$&$n=2k+1,k\not\equiv2\pmod{5}$& $(125,9),(245,11)$&Thm \ref{TH7}\\
	\hline   $x^{3(2^k-1)}$	&$n=2k+1,k\not\equiv13\pmod{27}$&$(45,9),(93,11)$&Thm \ref{TH3} \\
    \hline	$x^{3(2^{k+1}+1)}$&$n=3k+1,k\not\equiv9\pmod{14}$&$(51,10)$ &Thm \ref{TH5}\\
   	\hline   $x^{-9}$&$9\nmid n$& $(447,10)$& Thm \ref{TH1}  \\         
	\hline	
\end{tabular}
% The values of ($d,n$) just show the exponents of 0-APN (but not APN) power functions
%$x^d$ over $\mathbb{F}_{2^n}$ ($1\leq n\leq11$), which appeared in Table \ref{table1}.
\end{table}

%\begin{definition}(\cite{CCZ})\label{CCZ1}
%	The power functions $x^k$ and $x^l$ on $\mathbb{F}_{p^n}$ are CCZ equivalent, if and only if there exists a positive integer $0\leq t\leq n$, such that $l\equiv p^tk\pmod{p^n-1}$ or $kl\equiv p^t\pmod{p^n-1}$.
%\end{definition}

\section{On the resultant of polynomials}\label{Preliminaries}

In this section, in order to prove our results, we need concept and relevant conclusions of the resultant of two polynomials.

\begin{definition}(\cite{RH})
Let $\mathbb{K}$ be a field, $f(x)=a_{0}x^n+a_{1}x^{n-1}+\cdots+a_{n}\in\mathbb{K}[x]$ and $g(x)=b_{0}x^m+b_{1}x^{m-1}+\cdots+a_{m}\in\mathbb{K}[x]$ be two polynomials of degree $n$ and $m$ respectively, where $n,m\in\mathbb{N}$. Then the resultant $Res(f,g)$ of the two polynomials is defined by the determinant
\begin{equation*}
\begin{tikzpicture}
	\draw(0,0)node{$Res(f,g)=
		\left|
		\begin{array}{cccccccc}
			a_{0} & a_{1} & \cdots & a_{n} & 0&  & \cdots & 0 \\
			0 & a_{0} & a_{1} & \cdots & a_{n} & 0 & \cdots & 0 \\
			\vdots &  &   & &  &  &   & \vdots \\
			0 & \cdots& 0 & a_{0} & a_{1} &  & \cdots & a_{n} \\
			b_{0} & b_{1} & \cdots &  & b_{m}& 0 & \cdots & 0 \\
			0& b_{0} & b_{1} & \cdots &  & b_{m}&  \cdots & 0  \\
			\vdots &  &   &  &  &  &   & \vdots \\
			0 & \cdots & 0& b_{0} & b_{1} &  & \cdots & b_{m} \\
		\end{array}
		\right|
		$};
	\draw(3.6,1.1)node[right]{$\left.\rule{0mm}{11mm}\right\}m$ $rows$};
	\draw(3.6,-1.1)node[right]{$\left.\rule{0mm}{11mm}\right\}n$ $rows$};
\end{tikzpicture}
\end{equation*}
of order $m+n$.
	\end{definition}

If the degree of $f$ is $n$ and $f(x)=a_{0}(x-\alpha_{1})(x-\alpha_{2})\cdots(x-\alpha_{n})$ in the splitting field of $f$ over $\mathbb{K}$, then $Res(f,g)$ is also given by the formula
$$Res(f,g)=a_{0}^{m}\prod_{i=1}^{n}g(\alpha_{i}).$$
In this case, we have $Res(f,g)=0$ if and only if $f$ and $g$ have a common root, which means that $f$ and $g$ have a common divisor in $\mathbb{K}[x]$ of positive degree.

For two polynomials $F(x,y),G(x,y)\in\mathbb{K}[x,y]$ of positive degree in $y$, the resultant $Res(F,G,y)$ of $F$ and $G$ with respect to $y$ is the resultant of $F$ and $G$ when considered as polynomials in the single variable $y$. In this case, $Res(F,G,y)\in\mathbb{K}[x]\cap<F,G>$, where $<F,G>$ is the ideal generated by $F$ and $G$. Thus any pair $(a,b)$ with $F(a,b)=G(a,b)=0$ is such that $Res(F,G,y)(a)=0$.

\section{New classes of 0-APN power functions over $\mathbb{F}_{2^{n}}$}\label{Newclasses}

In this section, we give seven new classes of $0$-APN power functions $f(x)=x^d$ over $\mathbb{F}_{2^{n}}$. Let $d$ be uniquely determined by $k$. If $k$ and $n$ satisfy a linear relationship, we have the following two classes of $0$-APN power functions.

\begin{theorem}\label{TH6}
	Let $n$ and $k$ be positive integers with $n=2k+1$. Then 
	$f(x)=x^{3\cdot2^{k}-7}$
	is a $0$-APN function over $\mathbb{F}_{2^{n}}$.
\end{theorem}

\begin{proof} To show $f$ is $0$-APN, 
	it suffices to prove that the eqution
	\begin{equation}\label{3.6.1}
		(x+1)^{3\cdot2^{k}-7}+x^{3\cdot2^{k}-7}+1=0
	\end{equation}
	has no solution in $\mathbb{F}_{2^{n}}\setminus{\mathbb{F}_2}$. Assume that $x\in\mathbb{F}_{2^{n}}\setminus{\mathbb{F}_2}$ is a solution of (\ref{3.6.1}). Multiplying $x^7(x+1)^7$ on both sides of (\ref{3.6.1}), we have
	\begin{equation}\label{3.6.2}
		x^{2^{k+1}+9}+x^{2^{k}+9}+x^9+x^{3\cdot2^k+8}+x^{2^{k+1}+8}+x^{2^{k}+8}+x^{3\cdot2^k+1}+x^{16}=0.
	\end{equation}
	Let $y=x^{2^{k}}$, then $y^{2^{k+1}}=x$. (\ref{3.6.2}) can be written as
	\begin{equation}\label{3.6.3}
		x^9y^2+x^9y+x^9+x^8y^3+x^8y^2+x^8y+xy^3+x^{16}=0.
	\end{equation}
	Raising the $2^{k+1}$-th power on both sides of (\ref{3.6.3}), we have
	\begin{equation}\label{3.6.4}
		x^2y^{18}+xy^{18}+y^{18}+x^3y^{16}+x^{2}y^{16}+xy^{16}+x^3y^2+y^{32}=0.
	\end{equation}
	Computing the resultant of (\ref{3.6.3}) and (\ref{3.6.4}) with respect to $y$, and then decomposing it into the product of irreducible factors in $\mathbb{F}_{2}[x]$ as
	
	\begin{align}\label{3.6.5}
		&x^{57}(x+1)^{57}(x^2+x+1)^{10}(x^3+x+1)^4(x^3+x^2+1)^4(x^7+x^3+1)(x^7+x^4+1)\notag\\
		&(x^7+x^3+x^2+x+1)(x^7+x^6+x^5+x^3+x^2+x+1) \notag\\
		&(x^7+x^6+x^5+x^4+1)(x^7+x^6+x^5+x^4+x^2+x+1) \notag\\
		&(x^{26}+x^{16}+x^{15}+x^{14}+x^{13}+x^9+x^8+x^7+x^4+x^3+1)\notag\\
		&(x^{26}+x^{23}+x^{22}+x^{19}+x^{18}+x^{17}+x^{13}+x^{12}+x^{11}+x^{10}+1)\notag\\
		&(x^{26}+x^{23}+x^{22}+x^{21}+x^{18}+x^{16}+x^{15}+x^{13}+x^{12}+x^{7}+x^6+x^4+x^2+x+1)\notag\\
		&(x^{26}+x^{23}+x^{22}+x^{21}+x^{20}+x^{19}+x^{15}+x^{14}+x^{12}+x^{11}+x^{7}+x^6+x^3+x^2+1)\notag\\
		&(x^{26}+x^{24}+x^{18}+x^{15}+x^{12}+x^{11}+x^{10}+x^{9}+x^{6}+x^{5}+x^{4}+x^3+x^2+x+1)\notag\\
		&(x^{26}+x^{24}+x^{23}+x^{17}+x^{15}+x^{13}+x^{12}+x^{8}+x^{7}+x^{4}+x^2+x+1)\notag\\
		&(x^{26}+x^{24}+x^{23}+x^{20}+x^{19}+x^{15}+x^{14}+x^{12}+x^{11}+x^{7}+x^{6}+x^5+x^4+x^3+1)\notag\\
		&(x^{26}+x^{24}+x^{23}+x^{21}+x^{18}+x^{17}+x^{13}+x^{11}+x^{10}+x^{8}+x^{7}+x^3+x^2+x+1)\notag\\
		&(x^{26}+x^{25}+x^{24}+x^{22}+x^{19}+x^{18}+x^{14}+x^{13}+x^{11}+x^{9}+x^{3}+x^{2}+1)\notag\\
		&(x^{26}+x^{25}+x^{24}+x^{22}+x^{20}+x^{19}+x^{14}+x^{13}+x^{11}+x^{10}+x^{8}+x^{5}+x^4+x^3+1)\notag\\
		&(x^{26}+x^{25}+x^{24}+x^{23}+x^{19}+x^{18}+x^{16}+x^{15}+x^{13}+x^{9}+x^{8}+x^{5}+x^3+x^2+1)\notag\\
		&(x^{26}+x^{25}+x^{24}+x^{23}+x^{22}+x^{21}+x^{20}+x^{17}+x^{16}+x^{15}+x^{14}+x^{11}+x^8+x^2+1).
	\end{align}
	Note that $x\notin\mathbb{F}_{2}$, we assert that $x\in\mathbb{F}_{2^2}$, $x\in\mathbb{F}_{2^3}$, $x\in\mathbb{F}_{2^7}$  or $x\in\mathbb{F}_{2^{26}}\setminus\mathbb{F}_{2^{13}}$.
	
	Assume that $x\in\mathbb{F}_{2^2}$. We have $x\in\mathbb{F}_{2^2}\cap\mathbb{F}_{2^n}=\mathbb{F}_{2}$ since $n$ is odd, which is a contradiction.	
	
	Assume that $x\in\mathbb{F}_{2^3}$. If $k\not\equiv1\pmod3$, we have $x\in\mathbb{F}_{2^3}\cap\mathbb{F}_{2^n}=\mathbb{F}_{2}$, which is a contradiction. If $k\equiv1\pmod3$, then we have $x^{2^{k}}=x^2$ since $x\in\mathbb{F}_{2^3}$. Thus we derive from (\ref{3.6.2}) that $x^7(x+1)^7(x^2+x+1)=0$. This means $x\in\mathbb{F}_{2^2}$, and then $x\in\mathbb{F}_{2^3}\cap\mathbb{F}_{2^2}=\mathbb{F}_{2}$, which is a contradiction.
	
	Assume that $x\in\mathbb{F}_{2^7}$. when $k\not\equiv3\pmod7$, then $7\nmid n$, we have $x\in\mathbb{F}_{2^7}\cap\mathbb{F}_{2^n}=\mathbb{F}_{2}$, which is a  contradiction. If $k\equiv3\pmod7$, we have $x^{2^{k}}=x^8$ since $x\in\mathbb{F}_{2^7}$. Thus we derive from (\ref{3.6.2}) that $x^9(x+1)^9(x^2+x+1)(x^4+x+1)(x^4+x^3+1)(x^4+x^3+x^2+x+1)=0$. This means $x\in\mathbb{F}_{2^4}$, and then $x\in\mathbb{F}_{2^7}\cap\mathbb{F}_{2^4}=\mathbb{F}_{2}$, which is a contradiction.
	
	Assume that $x\in\mathbb{F}_{2^{26}}\backslash\mathbb{F}_{2^{13}}$. Then $x\in\mathbb{F}_{2^{26}}\cap\mathbb{F}_{2^{n}}=\mathbb{F}_{2}$ or $\mathbb{F}_{2^{13}}$ since $n$ is odd, which is a contradiction. Therefore (\ref{3.6.1}) has no solution in $\mathbb{F}_{2^{n}}\setminus{\mathbb{F}_2}$. The proof is completed.
\end{proof}

	\begin{theorem}\label{TH4}
	Let $n$ and $k$ be positive integers with $n=3k+1$. Then 
	$f(x)=x^{2^{2k+1}-2^{k+1}-2^k+1}$
	is a 0-APN function over $\mathbb{F}_{2^{n}}$.
\end{theorem}

\begin{proof}
	We need to prove that the equation
	\begin{equation}\label{3.3.1}
		(x+1)^{2^{2k+1}-2^{k+1}-2^k+1}+x^{2^{2k+1}-2^{k+1}-2^k+1}+1=0
	\end{equation}
	has no solution in $\mathbb{F}_{2^{n}}\setminus{\mathbb{F}_2}$. Assume that $x\in\mathbb{F}_{2^{n}}\setminus{\mathbb{F}_2}$ is a solution of (\ref{3.3.1}). Multiplying $x^{2^k+2^{k+1}}(x+1)^{2^k+2^{k+1}}$ on both sides of (\ref{3.3.1}), we have
	\begin{align}
		&x^{2^{2k+1}+2^{k+1}+2^{k}}+x^{2^{k+1}+2^{k}+1}+x^{2^{2k+1}+2^{k+1}+1}+x^{2^{2k+1}+2^{k}+1} \notag\\
		&+x^{2^{2k+1}+1}+x^{2^{k+2}+2^{k+1}}+x^{2^{k+2}+2^{k}}+x^{2^{k+2}}=0.\label{3.3.2}
	\end{align}
	Let $y=x^{2^{k}}$ and $z=y^{2^k}$,  then $z^{2^{k+1}}=x$. Raising the $2^k$-th and $2^{2k+1}$-th powers on  (\ref{3.3.2}) respectively, we obtain
	\begin{align}\label{3.3.3}
		\left\{
		\begin{array}{ll}
			y^3z^2+xy^3+xy^2z^2+xyz^2+xz^2+y^6+y^5+y^4=0	,\\
			xz^3+yz^3+xyz^2+xyz+xy+z^6+z^5+z^4=0	,\\
			x^3y^2+x^3z^2+x^2y^2z^2+xy^2z^2+y^2z^2+x^6+x^5+x^4=0.
		\end{array}
		\right.
	\end{align}
	Computing the resultants of the first and second equations of (\ref{3.3.3}) with respect to $z$, we have
	\begin{align*}
		Res_{1}(x,y) =& y^2(y+1)^2(x^8y^{12}+x^8y^8+x^8y^6+x^8y^4+x^8y^3+x^8+x^7y^{12}+x^{7}y^{10} \\
		&+x^7y^6+x^7y^4+x^6y^{18}+x^6y^{17}+x^6y^{16}+x^6y^{13}+x^6y^{11}+x^6y^8+x^6y^7 \\
		&+x^6y^6+x^5y^{20}+x^5y^{18}+x^5y^{14}+x^5y^{12}+x^4y^{21}+x^4y^{18}+x^4y^{17}+x^4y^{15}  \\
		&+x^4y^{14}+x^4y^{11}+x^3y^{20}+x^3y^{18}+x^3y^{14}+x^3y^{12}+x^2y^{26}+x^2y^{25}+x^2y^{24}  \\
		&+x^2y^{21}+x^2y^{19}+x^2y^{16}+x^2y^{15}+x^2y^{14}+xy^{28}+xy^{26}+xy^{22}+xy^{20}\\
		&+y^{32}+y^{29}+y^{28}+y^{26}+y^{24}+y^{23}+y^{20}).
	\end{align*}
	Similarly, computing the resultants of the first and third equations of (\ref{3.3.3}) with respect to $z$, we have
	\begin{align*}
		Res_{2}(x,y) =& 
		(x^7y^2+x^7y+x^7+x^6y^3+x^6y^2+x^6y+x^6+x^5y^3+x^5y^2+x^5y+x^5+x^4y^4 \\
		&+x^4y^3+x^4y^2+x^3y^6+x^3y^5+x^3y^4+x^2y^8+x^2y^7+x^2y^6+x^2y^5+xy^8\\
		&+xy^7+xy^6+xy^5+y^8+y^7+y^6)^2 .
	\end{align*}
	Computing the resultants of $Res_{1}(x,y)$ and $Res_{2}(x,y)$ with respect to $y$ by MAGMA computation, and then decomposing it into the product of irreducible factors in $\mathbb{F}_{2}[x]$ as	
	\begin{align}\label{3.3.4}
		&x^{124}(x+1)^{124}(x^2+x+1)^8(x^7+x^4+x^3+x^2+1)^2(x^7+x^5+x^2+x+1)^2(x^7+x^5 \notag\\
		&+x^4+x^3+1)^2(x^7+x^6+x^3+x+1)^2(x^7+x^6+x^4+x+1)^2(x^7+x^6+x^5+x^2+1)^2   \notag\\
		&(x^{15}+x^{10}+x^9+x^8+x^4+x^3+x^2+x+1)^4(x^{15}+x^{12}+x^9+x^8+x^5+x^4+x^2 \notag\\
		&+x+1)^2(x^{15}+x^{14}+x^{13}+x^{11}+x^{10}+x^7+x^6+x^3+1)^2(x^{15}+x^{14}+x^{13}\notag\\
		&+x^{12}+x^{11}+x^7+x^6+x^5+1)^4.
	\end{align}
	Thus we assert that  $x\in\mathbb{F}_{2^2}\setminus\mathbb{F}_2$, $x\in\mathbb{F}_{2^7}\setminus\mathbb{F}_2$ or $\mathbb{F}_{2^{15}}\setminus(\mathbb{F}_{2^5}\cup\mathbb{F}_{2^3})$. Suppose $x\in\mathbb{F}_{2^2}$. If $k$ is even, then $n$ is odd, we have $x\in\mathbb{F}_{2^2}\cap\mathbb{F}_{2^n}=\mathbb{F}_{2}$, which is a contradiction. If $k$ is odd, then $n$ is even, we have $x\in\mathbb{F}_{2^2}\cap\mathbb{F}_{2^n}=\mathbb{F}_{2^2}$, and consequently $x^{2^{k}}=x^2, x^{2^{2k}}=x$. Thus we derive from (\ref{3.3.2}) that $x^3(x+1)^3(x^3+x+1)(x^3+x^2+1)=0$. This means $x\in\mathbb{F}_{2^3}$, and then $x\in\mathbb{F}_{2^2}\cap\mathbb{F}_{2^3}=\mathbb{F}_{2}$, which is a contradiction.	
	
	Suppose $x\in\mathbb{F}_{2^7}$. If $k\not\equiv2\pmod7$, then $7\nmid n$, we have $x\in\mathbb{F}_{2^7}\cap\mathbb{F}_{2^n}=\mathbb{F}_{2}$, which is a contradiction. If $k\equiv2\pmod7$, we have $x^{2^{k}}=x^4$ and $x^{2^{2k}}=x^{16}$ since  $x\in\mathbb{F}_{2^7}$. Thus we derive from (\ref{3.3.2}) that $x^{13}(x+1)^{13}(x^6+x^3+1)(x^6+x^4+x^3+x+1)(x^6+x^5+x^3+x^2+1)=0$. This means $x\in\mathbb{F}_{2^6}$, and then $x\in\mathbb{F}_{2^6}\cap\mathbb{F}_{2^7}=\mathbb{F}_{2}$, which is a contradiction.
	
	Suppose $x\in\mathbb{F}_{2^{15}}\setminus(\mathbb{F}_{2^5}\cup\mathbb{F}_{2^3})$. Then $x\in\mathbb{F}_{2^{15}}\cap\mathbb{F}_{2^{n}}=\mathbb{F}_{2}$ or $\mathbb{F}_{2^5}$ since gcd$(n,3)=1$, which contradicts to $x\notin\mathbb{F}_{2^5}\cup\mathbb{F}_{2^3}$. Therefore (\ref{3.3.1}) has no solution in $\mathbb{F}_{2^{n}}\setminus{\mathbb{F}_2}$. The proof is completed.		
\end{proof}

%\begin{remark}
%  The values of ($d,n$) for $6\leq n\leq11$ in Thm \ref{TH1} are (27,6), (111,8) ,(447,10) under CCZ equivalence. When ($d,n$) take these values, power function $x^d$ over $\mathbb{F}_{2^n}$ is 0-APN but not APN. 
%\end{remark}
If $k$ (or equivalently, $n$) has some congruence relationship, we get four classes of $0$-APN power functions, which are listed as follows.

	\begin{theorem}\label{TH2}
	Let $n$ and $k$ be positive integers with $n=2k$ and $3\nmid k$. Then $f(x)=x^{3(2^{k}-1)}$	is a 0-APN function over $\mathbb{F}_{2^{n}}$.
\end{theorem}

\begin{proof}
	It suffices to show that the eqution
	\begin{equation}\label{3.5.1}
		(x+1)^{3(2^{k}-1)}+x^{3(2^{k}-1)}+1=0
	\end{equation}
	has no solution in $\mathbb{F}_{2^{n}}\setminus{\mathbb{F}_2}$. Suppose that $x\in \mathbb{F}_{2^{n}}\setminus{\mathbb{F}_2}$ is a solution of  (\ref{3.5.1}). Multiplying $x^3(x+1)^3$ on both sides of (\ref{3.5.1}), then (\ref{3.5.1}) becomes
	\begin{equation}\label{3.5.2}
		x^{2^{k+1}+3}+x^{2^{k}+3}+x^{3\cdot2^{k}+2}+x^{3\cdot2^{k}+1}+x^{3\cdot2^{k}}+x^6+x^5+x^4=0.
	\end{equation}
	Let $y=x^{2^{k}}$, then $y^{2^{k}}=x^{2^{2k}}=x$. Consequently, (\ref{3.5.2}) can be rewritten as
	\begin{equation}\label{3.5.3}
		x^3y^2+x^3y+x^2y^3+xy^3+y^3+x^6+x^5+x^4=0.
	\end{equation}
	Raising the $2^{k}$-th power on (\ref{3.5.3}) gives
	\begin{equation}\label{3.5.4}
		x^2y^3+xy^3+x^3y^2+x^3xy+x^3+y^6+y^5+y^4=0.
	\end{equation}
	Computing the resultant of (\ref{3.5.3}) and (\ref{3.5.4}) with respect to $y$, and the resultant can be decomposed into the product of irreducible factors in $\mathbb{F}_{2}[x]$ as
	\begin{equation}\label{3.5.5}
	Res(x)=x^9(x+1)^9(x^6+x^3+1)(x^6+x^4+x^3+x+1)(x^6+x^5+x^3+x^2+1)
	\end{equation}
	by MAGMA. Note that $x\notin \mathbb{F}_{2}$, $Res(x)=0$ implies $x^6+x^3+1=0$, $x^6+x^4+x^3+x+1=0$ or $x^6+x^5+x^3+x^2+1=0$. Without loss of generality, if $x$ satisfies $x^6+x^3+1=0$, $x\in\mathbb{F}_{2^{6}}\setminus(\mathbb{F}_{2^3}\cup\mathbb{F}_{2^2})$. Since $n=2k$ and $3\nmid k$,  $x\in\mathbb{F}_{2^6}\cap\mathbb{F}_{2^n}=\mathbb{F}_{2^2}$, which is a contradiction. This completes the proof.
\end{proof}

\begin{theorem}\label{TH7}
	Let $n$ and $k$ be positive integers with $n=2k+1$ and $k\not\equiv2\pmod{5}$. Then 
	$f(x)=x^{5(2^k+2^{k-1}+1)}$
	is a 0-APN function over $\mathbb{F}_{2^{n}}$.
\end{theorem}

\begin{proof}
	It suffices to show that the eqution
	\begin{equation}\label{3.7.1}
		(x+1)^{5(2^k+2^{k-1}+1)}+x^{5(2^k+2^{k-1}+1)}+1=0
	\end{equation}
	has no solution in $\mathbb{F}_{2^{n}}\setminus{\mathbb{F}_2}$. Assume that $x\in\mathbb{F}_{2^{n}}\setminus{\mathbb{F}_2}$ is a solution of (\ref{3.7.1}). Raising the square on (\ref{3.7.1}), we have
	\begin{equation}\label{3.7.2}
		(x^{10}+x^8+x^2+1)(x^{5\cdot2^{k+1}}+x^{2^{k+3}}+x^{2^{k+1}}+1)(x^{5\cdot2^{k}}+x^{2^{k+2}}+x^{2^{k}}+1)+x^{5(2^{k+1}+2^{k}+2)}+1=0.
	\end{equation}
	Let $y=x^{2^{k}}$, then $y^{2^{k+1}}=x$. (\ref{3.7.2}) can be written as
	\begin{equation}\label{3.7.3}
		(x^{10}+x^8+x^2+1)(y^{10}+y^8+y^2+1)(y^{5}+y^4+y+1)+x^{10}y^{15}+1=0.
	\end{equation}
	Raising the $2^{k+1}$-th power on both sides of (\ref{3.7.3}) gives
	\begin{equation}\label{3.7.4}
		(y^{20}+y^{16}+y^4+1)(x^{10}+x^8+x^2+1)(x^{5}+x^4+x+1)+x^{15}y^{20}+1=0.
	\end{equation}
	Computing the resultant of (\ref{3.7.3}) and (\ref{3.7.4}) with respect to $y$, and then decomposing it into the product of irreducible factors in $\mathbb{F}_{2}[x]$ as
	\begin{align}\label{3.7.5}
		&x^5(x+1)^5(x^2+x+1)^4(x^4+x+1)(x^4+x^3+1)(x^4+x^3+x^2+x+1)(x^5+x^2+1)  \notag\\
		&(x^5+x^3+1)(x^5+x^3+x^2+x+1) (x^5+x^4+x^2+x+1)(x^5+x^4+x^3+x+1) \notag\\
		&(x^5+x^4+x^3+x^2+1)(x^{14}+x^6+x^4+x^3+x^2+x+1)(x^{14}+x^9+x^8+x^5+x^4 \notag\\
		&+x^3+x^2+x+1)(x^{14}+x^{11}+x^{10}+x^5+x^4+x^3+x^2+x+1)(x^{14}+x^{11}\notag\\
		&+x^{10}+x^6+x^4+x^2+1)(x^{14}+x^{12}+x^{10}+x^8+x^4+x^3+1)(x^{14}+x^{12}\notag\\
		&+x^{10}+x^9+x^8+x^6+x^5+x^4+x^3+x^2+1)(x^{14}+x^{12}+x^{11}+x^{10}+x^9 \notag\\
		&+x^8+x^4+x^3+x^2+x+1)(x^{14}+x^{12}+x^{11}+x^{10}++x^9+x^8+x^6+x^5+x^4+x^2+1)\notag\\ &(x^{14}+x^{13}+x^{12}+x^{11}+x^{10}+x^6+x^5+x^4+x^3+x^2+1)(x^{14}+x^{13}+x^{12}+x^{11}   \notag\\
		&+x^{10}+x^8+1)(x^{14}+x^{13}+x^{12}+x^{11}+x^{10}+x^9+x^4+x^3+1)(x^{14}+x^{13}+x^{12}+x^{11}   \notag\\
		&+x^{10}+x^9+x^6+x^5+1)(x^{48}+x^{46}+x^{45}+x^{39}+x^{37}+x^{35}+x^{32}+x^{28}+x^{26}+x^{24}+x^{22}  \notag\\
		&+x^{20}+x^{16}+x^{14}+x^{11}+x^9+x^7+x^6+x^3+x^2+1)(x^{48}+x^{46}+x^{45}+x^{42}+x^{41}+x^{39}  \notag\\
		&+x^{37}+x^{34}+x^{32}+x^{28}+x^{26}+x^{24}+x^{22}+x^{20}+x^{16}+x^{13}+x^{11}+x^9+x^3+x^2+1)   \notag\\
		&(x^{48}+x^{47}+x^{45}+x^{44}+x^{42}+x^{41}+x^{40}+x^{37}+x^{36}+x^{33}+x^{31}+x^{30}+x^{29}+x^{27}+ \notag\\
		&x^{24}+x^{23}+x^{22}+x^{20}+x^{17}+x^{15}+x^{13}+x^{12}+x^{11}+x^{10}+x^9+x^7+x^5+x^4+x^2+x+1)  \notag\\
		&(x^{48}+x^{47}+x^{46}+x^{44}+x^{43}+x^{41}+x^{39}+x^{38}+x^{37}+x^{36}+x^{35}+x^{33}+x^{31}+x^{28}+\notag\\
		&x^{26}+x^{25}+x^{24}+x^{21}+x^{19}+x^{18}+x^{17}+x^{15}+x^{12}+x^{11}+x^{8}+x^7+x^6+x^4+x^3+x+1).
	\end{align}
	Note that $x\notin\mathbb{F}_{2}$, we assert that $x\in\mathbb{F}_{2^2}$, $x\in\mathbb{F}_{2^4}$, $x\in\mathbb{F}_{2^5}$, $x\in\mathbb{F}_{2^{14}}\backslash\mathbb{F}_{2^{7}}$ or $x\in\mathbb{F}_{2^{48}}\backslash\mathbb{F}_{2^{24}}$.
	
	Assume that $x\in\mathbb{F}_{2^2}$, We have $x\in\mathbb{F}_{2^2}\cap\mathbb{F}_{2^n}=\mathbb{F}_{2}$ since $n$ is odd, which is a contradiction. Similarly, we get $x\notin \mathbb{F}_{2^4}$.
	
	Assume that $x\in\mathbb{F}_{2^5}$. Note that $n=2k+1$ and $k\not\equiv2\pmod{5}$, then $5\nmid n$. We have $x\in\mathbb{F}_{2^5}\cap\mathbb{F}_{2^n}=\mathbb{F}_{2}$, which is a contradiction.
	
	Assume that  $x\in\mathbb{F}_{2^{14}}\backslash\mathbb{F}_{2^{7}}$. We have
	$\mathbb{F}_{2^{14}}\cap\mathbb{F}_{2^n}=\mathbb{F}_{2}$ or $\mathbb{F}_{2^7}$ since $n$ is odd, which is a contradiction. Similarly, we get $x\notin \mathbb{F}_{2^{48}}\backslash\mathbb{F}_{2^{24}}$. 
	Therefore (\ref{3.7.1}) has no solution in $\mathbb{F}_{2^{n}}\setminus{\mathbb{F}_2}$. The proof is completed.
\end{proof}	

\begin{theorem}\label{TH3}
	Let $n$ and $k$ be positive integers with $n=2k+1$ and $k\not\equiv13\pmod{27}$. Then $f(x)=x^{3(2^k-1)}$ is a 0-APN function over $\mathbb{F}_{2^{n}}$.
\end{theorem}

\begin{proof}
	It suffices to show that the eqution
	\begin{equation}\label{3.2.1}
		(x+1)^{3(2^k-1)}+x^{3(2^k-1)}+1=0
	\end{equation}
	has no solution in $\mathbb{F}_{2^{n}}\setminus{\mathbb{F}_2}$. Suppose that $x\in\mathbb{F}_{2^{n}}\setminus{\mathbb{F}_2}$ is a solution of (\ref{3.2.1}). Multiplying $x^3(x+1)^3$ on both sides of (\ref{3.2.1}), then (\ref{3.2.1}) becomes
	\begin{equation}\label{3.2.2}
		x^{2^{k+1}+3}+x^{2^k+3}+x^{3\cdot 2^k+2}+x^{3\cdot 2^k+1}+x^{3\cdot 2^k}+x^6+x^5+x^4=0.
	\end{equation}
	Raising the square on (\ref{3.2.2}) leads to 
	\begin{equation}\label{3.2.3}
		x^{2^{k+2}+6}+x^{2^{k+1}+6}+x^{3\cdot 2^{k+1}+4}+x^{3\cdot 2^{k+1}+2}+x^{3\cdot 2^{k+1}}+x^{12}+x^{10}+x^8=0.
	\end{equation}
	Let $y=x^{2^{k+1}}$, then $y^{2^{k}}=x$, (\ref{3.2.3}) can be written as
	\begin{equation}\label{3.2.4}
	x^6y^2+x^6y+x^4y^3+x^2y^3+y^3+x^{12}+x^{10}+x^8=0.
	\end{equation}
	Raising the $2^{k}$-th power on both sides of (\ref{3.2.4}) gives
	\begin{equation}\label{3.2.5}
		x^2y^3+xy^3+x^3y^2+x^3y+x^3+y^{6}+y^{5}+y^4=0.
	\end{equation}
	Computing the resultant of (\ref{3.2.4}) and (\ref{3.2.5}) with respect to $y$, and the resultant can be decomposed into the product of irreducible factors in $\mathbb{F}_{2}[x]$ as
	\begin{align}
		Res(x)=&x^{9}(x+1)^{9}(x^{27}+x^{22}+x^{20}+x^{19}+x^{15}+x^{13}+x^{12}+x^9+x^5+x^3+x^2+x+1)  \notag\\
		&(x^{27}+x^{26}+x^{25}+x^{24}+x^{22}+x^{18}+x^{15}+x^{14}+x^{12}+x^8+x^7+x^5+1)  \label{3.2.6} 
	\end{align}
by MAGMA. Note that $x\notin \mathbb{F}_{2}$, $Res(x)=0$ implies $x^{27}+x^{22}+x^{20}+x^{19}+x^{15}+x^{13}+x^{12}+x^9+x^5+x^3+x^2+x+1=0$ or $x^{27}+x^{26}+x^{25}+x^{24}+x^{22}+x^{18}+x^{15}+x^{14}+x^{12}+x^8+x^7+x^5+1=0$. Hence  $x\in\mathbb{F}_{2^{27}}\setminus\mathbb{F}_{2^9}$. Moreover, $x\in\mathbb{F}_{2^{27}}\cap\mathbb{F}_{2^n}=\mathbb{F}_{2^{\gcd(n,27)}}=\mathbb{F}_{2}$, $\mathbb{F}_{2^3}$  or $\mathbb{F}_{2^9}$ since $n$ is odd and $27\nmid n$, which contradicts to $x\notin\mathbb{F}_{2^9}$. Hence, (\ref{3.2.1}) has no solution in $\mathbb{F}_{2^{n}}\setminus{\mathbb{F}_2}$. We complete the proof.
	
\end{proof}

	\begin{theorem}\label{TH5}
		Let $n$ and $k$ be positive integers with  $n=3k+1$ and $k\not\equiv9\pmod{14}$. Then 
		$f(x)=x^{3(2^{k+1}+1)}$
		is a 0-APN function over $\mathbb{F}_{2^{n}}$.
	\end{theorem}
	
	\begin{proof}
		We only show that the eqution
	\begin{equation}\label{3.4.1}
		(x+1)^{3(2^{k+1}+1)}+x^{3(2^{k+1}+1)}+1=0
	\end{equation}
	has no solution in $\mathbb{F}_{2^{n}}\setminus{\mathbb{F}_2}$. Assume that $x\in\mathbb{F}_{2^{n}}\setminus{\mathbb{F}_2}$ is a solution of (\ref{3.4.1}), then $x$ satisfies
	\begin{align}
		&x^{2^{k+2}+3}+x^{2^{k+1}+3}+x^3+x^{3\cdot2^{k+1}+2}+x^{2^{k+2}+2}+x^{2^{k+1}+2}+x^2+x^{3\cdot2^{k+1}+1}  \notag\\
		&+x^{2^{k+2}+1}+x^{2^{k+1}+1}+x+x^{3\cdot2^{k+1}}+x^{2^{k+2}}+x^{2^{k+1}}=0.\label{3.4.2}
	\end{align}
		Let $y=x^{2^{k}}$ and $z=y^{2^k}$, then $z^{2^{k+1}}=x$. Raising the $2^k$-th and $2^{2k}$-th powers on both sides of (\ref{3.4.2}) respectively, we have 
\begin{align}\label{3.4.3}
	\left\{
	\begin{array}{ll}
		x^3y^4+x^3y^2+x^3+x^2y^6+x^2y^4+x^2y^2+x^2+xy^6+xy^4+xy^2+x+y^6+y^4+y^2=0,\\
		y^3z^4+y^3z^2+y^3+y^2z^6+y^2y^4+y^2z^2+y^2+yz^6+yz^4+yz^2+y+z^6+z^4+z^2=0,\\
		x^2z^3+xz^3+z^3+x^3z^2+x^2z^2+xz^2+z^2+x^3z+x^2z+xz+z+x^3+x^2+x=0.
	\end{array}
	\right.
\end{align}
Computing the resultant of the second and third equations of (\ref{3.4.3}) with respect to $z$,  we obtain
	\begin{align*}
	Res(x,y) =& (x^2+y)^3(x^{12}y^6+x^{12}y^5+x^{12}y^3+x^{12}y^2+x^{10}y^6+x^{10}y^5+x^{10}y^3+x^{10}y\\
	&+x^{8}y^4+x^{8}y^3+x^{8}+x^6y^6+x^6y^5+x^6y^4+x^6y^3+x^6y^2+x^6y+x^6+x^4y^6 \\
	&+x^4y^3+x^4y^2+x^2y^5+x^2y^3+x^2y+x^2+y^4+y^3+y+1).
\end{align*}	
Computing the resultant of	the first equation of (\ref{3.4.3}) and the above $Res(x,y)$ with respect to $y$, the resultant can be decomposed into the product of irreducible factors in $\mathbb{F}_{2}[x]$ as
\begin{align*}\label{3.4.4}
        &x^{5}(x+1)^{5}(x^2+x+1)^{3}(x^3+x+1)^{12}(x^3+x^2+1)^{12}(x^{14}+x^{10}+x^9\notag\\
        &+x^7+x^6+x^4+x^3+x+1)(x^{14}+x^{12}+x^9+x^8+x^7+x^6+x^5+x^2+1)(x^{14}+x^{13}\notag\\
        &+x^{11}+x^{10}+x^8+x^7+x^5+x^4+1).    
\end{align*}
Note that $x\notin\mathbb{F}_{2}$, we assert that $x\in \mathbb{F}_{2^2}$, $\mathbb{F}_{2^3}$ or 
$\mathbb{F}_{2^{14}}\backslash\mathbb{F}_{2^7}$.

Suppose $x\in \mathbb{F}_{2^2}$. If $k$ is even, then $n$ is odd, we have $x\in\mathbb{F}_{2^2}\cap\mathbb{F}_{2^n}=\mathbb{F}_{2}$, which contradicts to $x\notin\mathbb{F}_{2}$. If $k$ is odd, we have $x^{2^{k+1}}=x$ since $x\in \mathbb{F}_{2^2}$. Thus we derive from (\ref{3.4.1}) that $x^2(x+1)^2=0$. This means $x\in\mathbb{F}_{2}$, which is a contradiction.	

Suppose $x\in \mathbb{F}_{2^3}$. We have $x\in\mathbb{F}_{2^3}\cap\mathbb{F}_{2^n}=\mathbb{F}_{2}$ since $n=3k+1$, which is a contradiction.

Suppose that $x\in\mathbb{F}_{2^{14}}\backslash\mathbb{F}_{2^7}$. Since $n=3k+1$ and $k\not\equiv9\pmod{14}$, $14\nmid n$, and consequently $x\in\mathbb{F}_{2^{14}}\cap\mathbb{F}_{2^{n}}=\mathbb{F}_{2}$ or $\mathbb{F}_{2^7}$, which is a contradiction. Therefore  (\ref{3.4.1}) has no solution in $\mathbb{F}_{2^{n}}\setminus{\mathbb{F}_2}$. The proof is completed.
	
 \end{proof}

At the end of this section, we propose a class of $0$-APN power function over $\mathbb{F}_{2^{n}}$ that $d$ is a constant.

\begin{theorem}\label{TH1}
	The power function $f(x)=x^{-9}$ is a 0-APN function over $\mathbb{F}_{2^{n}}$ if and only if $9\nmid n$.
\end{theorem}

\begin{proof}  It suffices to prove that the eqution
	\begin{equation}\label{3.1.1}
		(x+1)^{-9}+x^{-9}+1=0
	\end{equation}
	has no solution in $\mathbb{F}_{2^{n}}\setminus{\mathbb{F}_2}$ when $9\nmid n$. Suppose that $x\in \mathbb{F}_{2^{n}}\setminus{\mathbb{F}_2}$ is a solution of (\ref{3.1.1}). Multiplying $x^9(x+1)^9$ on both sides of (\ref{3.1.1}), we obtain
	\begin{equation}\label{3.1.2}
		x^9+(x+1)^9+x^9(x+1)^9=0. 
	\end{equation}
	It is not difficult to see that (\ref{3.1.2}) can be decomposed as
	\begin{equation}\label{3.1.3}
		(x^9+x+1)(x^9+x^8+1)=0, 
	\end{equation}
	where $x^9+x+1$ and $x^9+x^8+1$ are irreducible polynomials over $\mathbb{F}_2$. Without loss of generality, suppose $x^9+x+1=0$, then $x\in\mathbb{F}_{2^9}\setminus\mathbb{F}_{2^3}$. Since $x\in\mathbb{F}_{2^9}$ and $9\nmid n$,  $x\in\mathbb{F}_{2^9}\cap \mathbb{F}_{2^n}=\mathbb{F}_{2^{(9,n)}}=\mathbb{F}_{2^3}$ or $\mathbb{F}_{2}$. We assert that such $x$ do no exist since $x\notin\mathbb{F}_{2^3}$. Therefore, (\ref{3.1.1}) has no solution in $\mathbb{F}_{2^{n}}\setminus{\mathbb{F}_2}$. Moreover, it is not difficult to see that (\ref{3.1.3}) has solution $x\neq 0, 1$ when $9|n$. The proof is completed.
\end{proof}

%\section{Conclusion}\label{conclusion}
%In this paper, based on the multivariate method and resultant elimination, we proposed several new infinite classes of 0-APN power functions over $\mathbb{F}_{2^{n}}$. Under the CCZ-equivalence \cite{CCZ}, these infinite classes of 0-APN power functions can explain some exponents for $1\leq n\leq11$ which are not yet  ``explained"  of Table \ref{table2} .

%\section*{Acknowledgments}

%\bibliographystyle{AIMS}
%\bibliography{refer}

\end{document}